\documentclass[11pt]{article}

\usepackage[]{inputenc}
\usepackage[T1]{fontenc}
\usepackage{hyperref}

\usepackage{color}

\usepackage{fullpage}

\usepackage{amsmath}
\usepackage{amssymb}
\usepackage[toc,page]{appendix}

\usepackage{amsthm}
\usepackage{float}
\usepackage{wrapfig}
\usepackage{caption}
\usepackage{graphicx}
\usepackage{pgfplots}
\usepackage{subcaption}

\def \be {\begin{equation}}
\def \ee {\end{equation}}
\def \bea {\begin{eqnarray}}
\def \eea {\end{eqnarray}}

\def \u{\underline}

\title{A third law of black hole mechanics for supersymmetric black holes and a quasi-local mass-charge inequality}

\def\be{\begin{equation}}
\def\ee{\end{equation}}
\def\ba{\begin{eqnarray}}
\def\ea{\end{eqnarray}}

\newtheorem{definition}{Definition}[section]
\newtheorem{theorem}{Theorem}[section]
\newtheorem{corollary}{Corollary}[theorem]
\newtheorem{lemma}[theorem]{Lemma}

\author{Harvey S. Reall \\ {\footnotesize Department of Applied Mathematics and Theoretical Physics, University of Cambridge} \\ {\footnotesize Wilberforce Road, Cambridge CB3 0WA, United Kingdom}\\ {\footnotesize hsr1000@cam.ac.uk}}

\begin{document}

\maketitle

\begin{abstract}
It has recently been proved that a third law of black hole mechanics does not hold for Einstein-Maxwell theory coupled to a massless charged scalar field: there exist solutions that describe gravitational collapse to form an exactly extremal Reissner-Nordstr\"om black hole in finite time. In this paper it is proved that such solutions do not exist in theories with matter fields satisfying a local mass-charge inequality. In such a theory, if a 2-surface has the same metric, extrinsic curvature, and Maxwell field as a cross-section of an extremal Reissner-Nordstr\"om horizon then this surface cannot have a compact interior and so cannot be a horizon cross-section of a black hole formed in gravitational collapse. This result is proved using spinorial techniques, which are also used to prove a mass-charge inequality for a modified version of the Dougan-Mason quasi-local mass.
\end{abstract}

\section{Introduction}

In recent work, Kehle and Unger \cite{Kehle:2022uvc} have proved that there exist solutions of Einstein-Maxwell theory coupled to a massless charged scalar field that describe spherically symmetric gravitational collapse to form a black hole that is {\it exactly} extremal Reissner-Nordstr\"om black hole after a finite advanced time. They also prove existence of solutions of this type for which there is an earlier period of advanced time during which the solution is {\it exactly} Schwarzschild on and inside the event horizon, an example of a non-extremal black hole that becomes extremal in finite time. See Fig. \ref{fig:KU}. The degree of differentiability can be made arbitrarily high, in contrast to earlier examples involving thin shells of matter \cite{israel:1967,kuchar:1968,Boulware:1973tlq}. This work proves that a ``third law of black hole mechanics'' \cite{Bardeen:1973gs} does not hold in this theory, contradicting an earlier claim \cite{Israel:1986gqz}. Kehle and Unger have proved similar results with the massless charged scalar replaced by charged Vlasov matter with small mass parameter \cite{Kehle:2024vyt}.

\begin{figure}
\centering
\includegraphics[width=8.0cm]{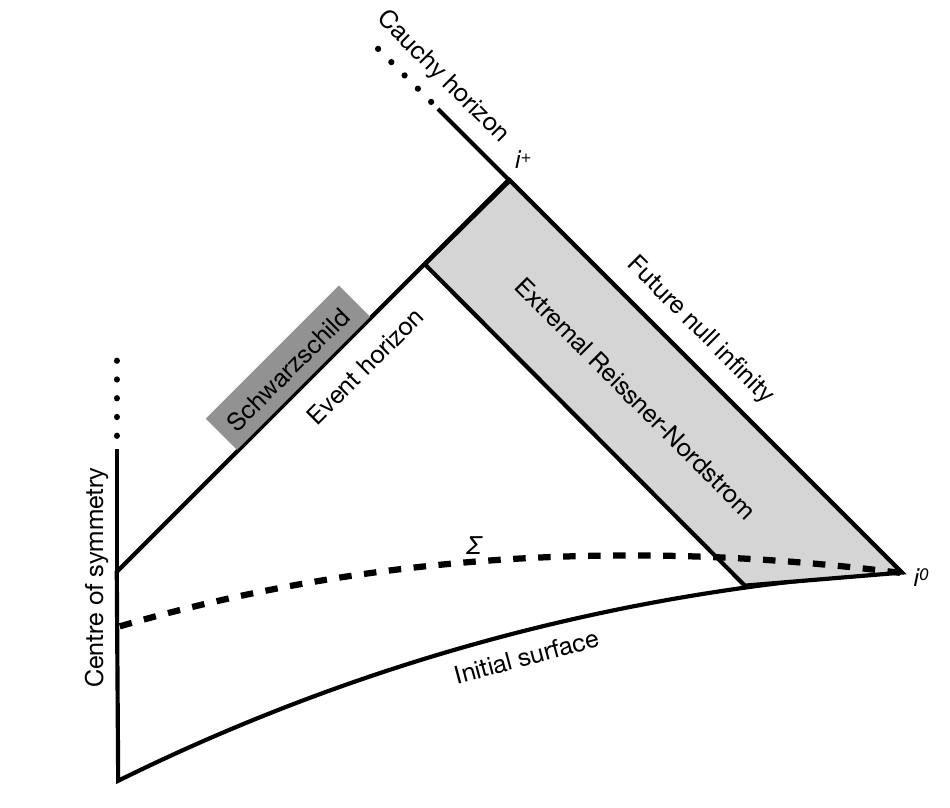}
\caption{Penrose diagram for the solutions constructed by Kehle and Unger \cite{Kehle:2022uvc} describing the formation of an exactly extremal Reissner-Nordstr\"om black hole in spherically symmetric gravitational collapse of a massless charged scalar field. The asymptotically flat initial surface describes the spacetime before collapse. The metric is exactly Schwarzschild in the dark shaded region.}
\label{fig:KU}
\end{figure}

\begin{figure}
\begin{subfigure}{0.5\textwidth}
\centering
\includegraphics[width=8.0cm]{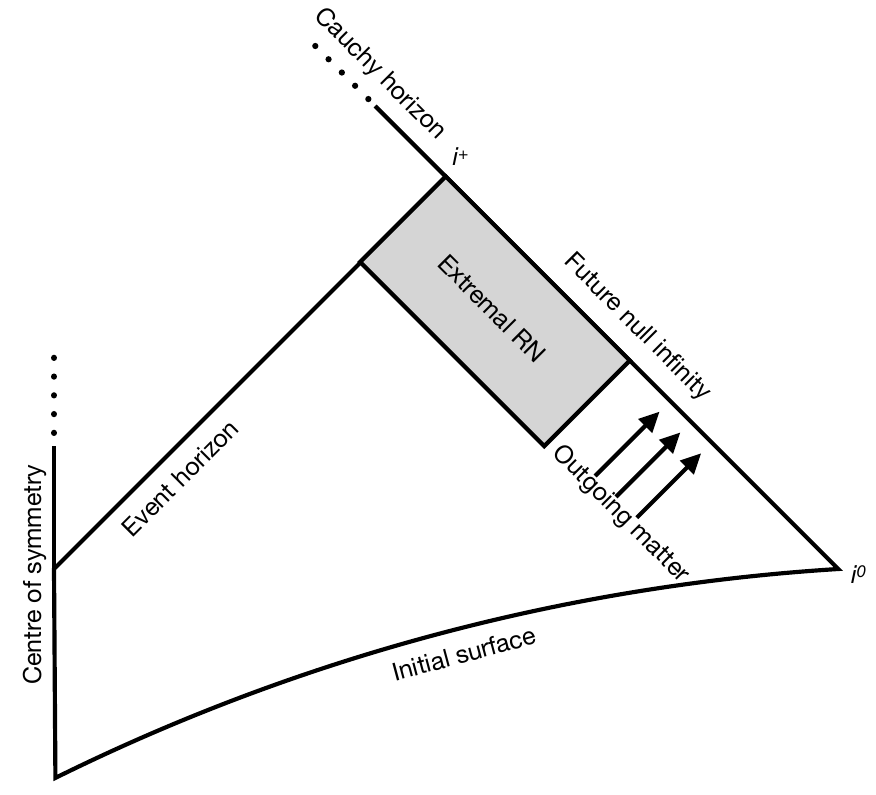}
\caption{}
\end{subfigure}
\begin{subfigure}{0.5\textwidth}
\centering
\includegraphics[width=8.0cm]{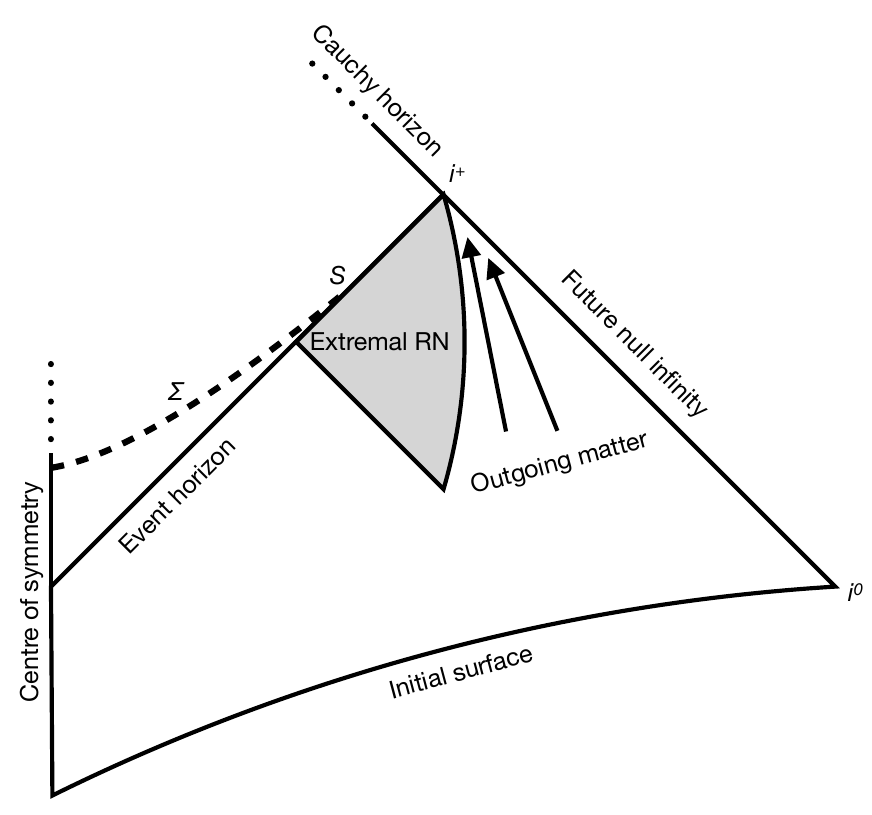}
\caption{}
\end{subfigure}
\caption{(a) A spacetime describing formation of an extremal Reissner-Nordstr\"om black hole inside an expanding bubble of massless charged scalar field. This spacetime has $M>|Q|$ at spatial infinity. (b) In a theory with both positively and negatively charged massive Vlasov matter, gravitational collapse might produce a positively charged extremal Reissner-Nordstr\"om black hole inside an expanding bubble, with negatively charged matter outside the bubble. This spacetime has $M>|Q|$ everywhere at infinity. The horizon cross-section $S$ is the boundary of a compact spacelike surface $\Sigma$.}
\label{fig:penrose}
\end{figure}

These results hold for theories where the matter has a large charge to mass ratio. One might expect it to be harder to form an extremal Reissner-Nordstr\"om black hole in a theory with matter obeying an upper bound on its charge to mass ratio. For example, if matter fields satisfy a certain local mass-charge inequality (a strengthened version of the dominant energy condition) then a spacetime that is smooth on a spacelike Cauchy surface $\Sigma$ as in Fig. \ref{fig:KU} must satisfy $M \ge \sqrt{Q^2+P^2}$, where $M$, $Q$ and $P$ are the mass, electric charge  and magnetic charge measured at spatial infinity, with equality if, and only if, the spacetime is {\it supersymmetric}, i.e., it admits a ``supercovariantly constant'' spinor \cite{Gibbons:1982fy,Gibbons:1982jg}. For such matter, this excludes spacetimes similar to those of Kehle and Unger, which have $M=|Q|$ (and $P=0$), because the existence of a supercovariantly constant spinor implies that the spacetime admits a causal Killing vector \cite{Gibbons:1982fy}, so it cannot describe gravitational collapse. 

This argument excludes (for suitable matter) the existence of solutions describing gravitational collapse to form a spacetime that is exactly extremal Reissner-Nordstr\"om all the way out to spatial infinity. However, it does not exclude the existence of solutions describing gravitational collapse to form a black hole that is exactly extremal Reissner-Nordstr\"om only near the horizon, but not near spatial infinity, which would still be a violation of the third law. 
For example, spherically symmetric gravitational collapse of a massless charged scalar might produce a spacetime with Penrose diagram as shown in Fig. \ref{fig:penrose}(a). After a finite advanced time, the spacetime is exactly extremal Reissner-Nordstr\"om inside an expanding ``bubble'' region. In the region outside the bubble there is a purely outgoing flux of scalar field (with charge opposite to that of the hole), the solution is not Reissner-Nordstr\"om, and $M>|Q|$ at infinity.  For the types of matter considered by Kehle and Unger one could construct such solutions by ``solving backwards in time'' from data prescribed on future null infinity and the future event horizon \cite{ryan_email}. This kind of spacetime requires massless charged matter but one can envisage the existence of similar solutions with massive charged Vlasov matter, as shown in Fig. \ref{fig:penrose}(b). More generally, there might exist solutions that, after a finite advanced time, are exactly extremal Reissner-Nordstr\"om on the event horizon but differ from extremal Reissner-Nordstr\"om immediately outside the horizon. Could such solutions exist in a theory with matter satisfying the local mass-charge inequality? They are not excluded by the argument of the previous paragraph because they have $M>|Q|$. Furthermore, since such matter is massive it cannot reach future null infinity and so one cannot simply reformulate the argument of the previous paragraph in terms of the mass and charge at null infinity. 

The aim of this paper is to prove that a third law of black hole mechanics does hold for {\it supersymmetric} black hole solutions of Einstein-Maxwell theory coupled to charged matter obeying the local mass-charge inequality of \cite{Gibbons:1982fy,Gibbons:1982jg}. A supersymmetric black hole is one admitting a non-trivial supercovariantly constant spinor, e.g. extremal Reissner-Nordstr\"om. The third law-violating spacetimes discussed in the previous paragraph are not supersymmetric but, at late time, the horizon is ``the same as'' the horizon of a supersymmetric black hole. We will prove this is impossible in theories obeying the local mass-charge inequality: if a black hole solution admits a horizon cross-section $S$ on which the (spacetime) metric, extrinsic curvature and Maxwell field coincide with those on a cross-section of the horizon of a supersymmetric black hole then $S$ {\it cannot have compact interior}. More precisely, there does not exist a smooth compact spacelike surface $\Sigma$ with $S = \partial \Sigma$ (see Fig. \ref{fig:penrose}(b)). Since a black hole formed in gravitational collapse {\it would} have compact interior, this proves that one cannot form a supersymmetric black hole in gravitational collapse in this class of theories. 

Our proof is based on ideas of Dougan and Mason \cite{Dougan:1991zz}, who introduced spinorial definitions of {\it quasilocal} $4$-momentum and mass for a $2$-surface $S=\partial \Sigma$, and proved that their definition satisfies a positive energy theorem. A supercovariant generalization of their methods enables our version of the third law to be proved fairly easily. The proof does not require the introduction of a quasilocal $4$-momentum or energy. However, for completeness, we go on to explain how to define notions of quasilocal $4$-momentum and mass $\varpi$ associated with $S$ and show that $\varpi \ge \sqrt{Q^2 + P^2}$ (where $Q$ and $P$ are the total electric and magnetic charges on $\Sigma$) with equality only if there exists a supercovariantly constant spinor on $\Sigma$. The mass $\varpi$ is a simple modification of the quasilocal mass defined by Dougan and Mason and agrees with the ADM mass when evaluated for a symmetry sphere in the Reissner-Nordstr\"om spacetime.  

This paper is organized as follows. In section \ref{sec:third} we present our proof of a third law for supersymmetric black holes. The main result is presented in Theorem \ref{thm:spinor} and Corollary \ref{cor:3rdlaw}. Section \ref{sec:quasilocal} describes our modification of the Dougan-Mason quasilocal $4$-momentum and mass. The main result is the mass-charge inequality established inTheorem \ref{thm:masscharge}. Section \ref{sec:discuss} contains a discussion of our results.  

\subsection*{Notation and conventions}

We set Newton's constant $G=1$. Lower case Latin indices denote abstract tensor indices. Upper case Latin indices are 2-component spinor indices. The use of 2-component spinors makes it convenient to follow the conventions of \cite{Penrose:1985bww}. In particular we shall use a negative signature metric. 
 
 \section{Third law}

\label{sec:third}

\subsection{Preliminaries}

We shall consider Einstein-Maxwell theory coupled to charged matter in four spacetime dimensions. The matter has energy-momentum tensor $T^{(m)}_{ab}$, electric current $J_a$ and magnetic current $\tilde{J}_a$. The Einstein equation is
\be
\label{einstein}
 R_{ab} - \frac{1}{2} R g_{ab} =   8 
 \pi  \left( T_{ab}^{(Max)} + T_{ab}^{(m)} \right)
\ee
where (recall we use a negative signature metric)
\be
  T_{ab}^{(Max)} =- \frac{1}{4\pi}  \left( F_a{}^c F_{bc} - \frac{1}{4} F_{cd}F^{cd} g_{ab} \right)
\ee
and the Maxwell equations are
\be
 d\star F= 4\pi \star J \qquad \qquad dF = 4\pi \star \tilde{J}
\ee
The matter is said to satisfy the {\it local mass-charge inequality} if the following strengthened version of the dominant energy condition holds w.r.t. any orthonormal basis $\{e^a_\mu\}$ \cite{Gibbons:1982fy,Gibbons:1982jg}: 
\be
 T^{(m)}_{00} \ge \sqrt{T^{(m)}_{0i}T^{(m)}_{0i} + J_0^2 + \tilde{J}_0^2}
\ee
If we set $e_0^a = Z^a$ then this is equivalent to
\be
\label{local_mass_charge}
 T_{ab}^{(m)} Z^a Z^b \ge 0 \qquad \qquad T^{(m)}_{ab} Z^b T^{(m)a}{}_c Z^c \ge (J_a Z^a)^2 + (\tilde{J}_a Z^a)^2
\ee
and, by rescaling, this must hold for any timelike $Z^a$ and by continuity also for null $Z^a$. Hence the matter satisfies the local mass-charge inequality if, and only if, \eqref{local_mass_charge} holds for any causal vector $Z^a$. We have not included a cosmological constant in \eqref{einstein} but a positive cosmological constant could be accommodated by absorbing it into $T_{ab}^{(m)}$, which will still satisfy \eqref{local_mass_charge}. 

The electric and magnetic charges enclosed by a $2$-surface $S$ are defined by
\be
\label{QPdef}
 Q + iP = \frac{1}{4\pi} \int_S (\star F + i F) 
\ee
The supercovariant derivative acting on a Dirac spinor $\epsilon$ is defined by \cite{Gibbons:1982fy}
\be
 \hat{\nabla}_a \epsilon = \nabla_a \epsilon + \frac{1}{4} F_{bc} \gamma^b \gamma^c \gamma_a \epsilon
\ee
where $\gamma^a$ are the Dirac gamma matrices satisfying $\{ \gamma^a ,\gamma^b\} = 2g^{ab}$ with $\gamma^0$ hermitian and $\gamma^{1,2,3}$ anti-hermitian.  A Dirac spinor can be decomposed into a pair of 2-component spinors as $\epsilon=(\bar{\lambda}^A,\mu_{A'})$. The supercovariant derivative decomposes as \cite{Gibbons:1982jg,Tod:1983pm,Chrusciel:2005ve},
\be
\label{supercov}
\hat{\nabla}_{AA'} \lambda_{B'} \equiv \nabla_{AA'} \lambda_{B'} + \sqrt{2} \bar{\phi}_{A'B'} \bar{\mu}_A \qquad \qquad
 \hat{\nabla}_{AA'} \mu_{B'} \equiv  \nabla_{AA'} \mu_{B'} - \sqrt{2} \bar{\phi}_{A'B'} \bar{\lambda}_A
\ee
where $\phi_{AB}$ is the symmetric spinor describing the Maxwell field:
\be
 F_{ab} = \phi_{AB} \epsilon_{A'B'} + \bar{\phi}_{A'B'} \bar{\epsilon}_{AB}
\ee
From a spinor $\epsilon$ we can define a vector and a complex scalar \cite{Tod:1983pm} (we adopt the normalization of \cite{Chrusciel:2005ve})
\be
 \label{XVdef}
 X^a = \frac{1}{\sqrt{2}} \left( \bar{\lambda}^A \lambda^{A'} + \bar{\mu}^A \mu^{A'} \right) \qquad V = \bar{\lambda}_A \bar{\mu}^A
\ee
which in Dirac notation correspond to the vector $\bar{\epsilon} \gamma^a \epsilon$ and two real scalars $\bar{\epsilon} \epsilon$, $\bar{\epsilon} \gamma_5 \epsilon$. These quantities satisfy
\be
\label{Xnorm}
 X_a X^a = V \bar{V} 
\ee
and so $X^a$ is non-spacelike (and vanishes iff $\epsilon$ vanishes). 

Throughout this paper, we shall assume that the spacetime and Maxwell field are smooth. This is simply for convenience, and it should be straightforward to adapt our results to non-smooth spacetimes which exceed some minimum level of differentiability. This point is worth emphasizing since the scalar field examples of Kehle and Unger are not smooth \cite{Kehle:2022uvc} (their Vlasov examples {\it are} smooth \cite{Kehle:2024vyt}). However, their method can be used to construct third-law-violating spacetimes of any desired degree of differentiability. The finitely differentiable version of our results would apply to spacetimes with the same level of differentiability as ``most'' of these examples.  

\subsection{Supersymmetric surfaces}

The extremal Reissner-Nordstr\"om solution (or an analytic extension of it) is {\it supersymmetric}: it admits a non-trivial globally defined supercovariantly constant spinor field. A spacetimes of the form shown in Fig. \ref{fig:KU} or \ref{fig:penrose} is not supersymmetric: a supercovariantly constant spinor exists in the extremal Reissner-Nordstr\"om region but it cannot be smoothly extended to the entire spacetime. A late-time cross-section of the horizon of such a black hole would be an example of a {\it supersymmetric surface}, which we define as follows.

\begin{definition}
Let $S$ be a smooth connected spacelike 2-surface in a smooth spacetime containing a Maxwell field. $S$ is a {\rm supersymmetric surface} if there exists a Dirac spinor field $\epsilon$ defined on $S$ such that $t^a \hat{\nabla}_a \epsilon=0$ for any vector field $t^a$ tangent to $S$, with $\epsilon$ not identically zero on $S$.
\end{definition}
 
 A standard property of such a spinor is the following:
 
 \begin{lemma}
\label{lem:nonzero}
  If $S$ is a supersymmetric surface then the spinor $\epsilon$ is non-zero everywhere on $S$.
 \end{lemma}
 \begin{proof}
By definition, $\epsilon \ne 0$ at some point $p \in S$. Assume that $\epsilon$ vanishes at $q \in S$ and let $C$ be a smooth curve in $S$ connecting $q$ to $p$. Let $t^a$ be tangent to $C$. Then $t^a \hat{\nabla}_a \epsilon=0$ is a homogeneous linear first order ODE along $C$ with initial condition $\epsilon=0$ at $q$. The unique solution is $\epsilon=0$, which contradicts $\epsilon(p) \ne 0$. 
\end{proof} 

 Consider a supersymmetric spacetime, i.e., one admitting a supercovariantly constant spinor $\epsilon$ that is not identically zero. Arguing similarly to Lemma \ref{lem:nonzero} proves that $\epsilon$ is everywhere non-zero. The extremal Reissner-Nordstr\"om solution (or its analytic extension) is an example of such a spacetime. Now let $S$ be any smooth surface in such a spacetime. Clearly $t^a \hat{\nabla}_a \epsilon=0$ on $S$ so we have shown that $S$ is a supersymmetric surface. Supersymmetric surfaces can also exist in a non-supersymmetric spacetime, i.e., one not admitting a globally defined supercovariantly constant spinor. In the solutions constructed by Kehle and Unger (Fig. \ref{fig:KU}), any spacelike 2-surface $S$ in the extremal Reissner-Nordstr\"om region is supersymmetric but the interior of such a surface is not supersymmetric, i.e., if we write $S=\partial \Sigma$ with $\Sigma$ a compact spacelike surface then there does not exist a supercovariantly constant spinor defined everywhere on $\Sigma$. We shall prove that this behaviour cannot occur in a theory satisfying the local mass-charge inequality \eqref{local_mass_charge} and explain why this implies that the third law cannot be violated by a supersymmetric black hole in such a theory. 

We shall use an approach based on work of Dougan and Mason \cite{Dougan:1991zz}, who showed how to define a quasilocal $4$-momentum for a $2$-surface $S$ and proved that their definition satisfies a positive energy theorem. Our approach is essentially a supercovariant version of theirs. Some of the steps involved in adapting their work to the supercovariant setting have already been made in \cite{Rogatko:2002qe}, in particular Theorem \ref{thm:Ipos} below is similar to a result in  \cite{Rogatko:2002qe}.

Let $S$ be any smooth compact spacelike $2$-surface in a smooth spacetime containing a smooth Maxwell field. Introduce a null tetrad $\{ \ell^a,n^a,m^a,\bar{m}^a\}$ where, on $S$, $\ell^a$ and $n^a$ are future-directed null normals to $S$, and hence $m^a$ and $\bar{m}^a$ are tangent to $S$. Introduce a spinor basis $\{o^A,\iota^A \}$ such that
\be
\ell^a = o^A \bar{o}^{A'}, \qquad n^a = \iota^A \bar{\iota}^{A'}, \qquad m^a = o^A \bar{\iota}^{A'} \qquad \bar{m}^A = \iota^A \bar{o}^{A'}
\ee
The condition for $S$ to be supersymmetric is the existence of non-zero $\epsilon$ such that $m^a \hat{\nabla}_a \epsilon= \bar{m}^a \hat{\nabla}_a \epsilon=0$. From \eqref{supercov}, this is equivalent to the following equations:
\begin{subequations}
\be
\label{hol_conds1}
  \bar{\eth} \lambda_{1'} + \sqrt{2} \bar{\phi}_{0'1'} \bar{\mu}_1 + \rho' \lambda_{0'}= 0 \qquad \qquad \bar{\eth} \mu_{1'} -  \sqrt{2} \bar{\phi}_{0'1'} \bar{\lambda}_1 + \rho' \mu_{0'}= 0
\ee
\be
 \label{hol_conds2}
\bar{\eth} \lambda_{0'} + \sqrt{2} \bar{\phi}_{0'0'} \bar{\mu}_1 + \bar{\sigma} \lambda_{1'} = 0 \qquad  \qquad \bar{\eth} \mu_{0'} - \sqrt{2} \bar{\phi}_{0'0'} \bar{\lambda}_1+ \bar{\sigma} \mu_{1'} = 0
\ee
\end{subequations}
\begin{subequations}
\be
\label{anti_hol_conds1}
 \eth \lambda_{0'} + \sqrt{2} \bar{\phi}_{0'1'} \bar{\mu}_0 + \rho \lambda_{1'} = 0  \qquad \qquad  \eth \mu_{0'} - \sqrt{2} \bar{\phi}_{0'1'} \bar{\lambda}_0 +  \rho \mu_{1'} = 0
 \ee
\be
 \label{anti_hol_conds2}
 \eth \lambda_{1'} + \sqrt{2} \bar{\phi}_{1'1'} \bar{\mu}_0 + \bar{\sigma}' \lambda_{0'} = 0  \qquad \qquad  \eth \mu_{1'} - \sqrt{2} \bar{\phi}_{1'1'} \bar{\lambda}_0 +  \bar{\sigma}' \mu_{0'} = 0
\ee
\end{subequations}
These equations involve the GHP derivative operator $\eth$ \cite{Geroch:1973am}. The Newman-Penrose coefficients $\rho$, $\sigma$ ($\rho'$, $\sigma'$) describe the expansion and shear of the null geodesics with tangent vector $\ell^a$ ($n^a$) on $S$. In particular $\rho,\rho'$ are real \cite{Geroch:1973am}. For a ``normal'' surface the outwards-directed null geodesics are expanding and the inwards-directed null geodesics are contracting. If $\ell^a$ points outwards then this corresponds to $\rho<0$, $\rho'>0$. A trapped surface has $\rho>0$, $\rho'>0$.

Following \cite{Gibbons:1982fy} we introduce the supercovariant generalization of the Nester-Witten \cite{Witten:1981mf,Nester:1981bjx} $2$-form:
\be
\label{hatLambda}
 \hat{\Lambda}_{AA'BB'} ={\rm Re} \left[ (-i) \left( \bar{\lambda}_A \hat{\nabla}_{BB'} \lambda_{A'} + \bar{\mu}_A \hat{\nabla}_{BB'} \mu_{A'} -  \bar{\lambda}_B \hat{\nabla}_{AA'} \lambda_{B'}  - \bar{\mu}_B \hat{\nabla}_{AA'} \mu_{B'} \right)\right]
\ee
For our $2$-surface $S$ and any Dirac spinor $\epsilon=(\bar{\lambda}^A,\mu_{A'})$ defined on $S$ we define the functional
\be
 \hat{I}_S[\epsilon] = \int_S \hat{\Lambda}
\ee
A calculation gives the explicit formula
\bea
\label{hatIS}
 \hat{I}_S[\epsilon] &=&{\rm Re}\, \int_S \left[ \bar{\lambda}_1 \left( \eth \lambda_{0'} +\sqrt{2} \bar{\phi}_{0'1'} \bar{\mu}_0 + \rho \lambda_{1'} \right) - \bar{\lambda}_0 \left( \bar{\eth} \lambda_{1'} +\sqrt{2} \bar{\phi}_{0'1'} \bar{\mu}_1+\rho' \lambda_{0'} \right) \right. \nonumber \\
 &+&  \left. \bar{\mu}_1 \left( \eth \mu_{0'} -\sqrt{2} \bar{\phi}_{0'1'} \bar{\lambda}_0 + \rho \mu_{1'} \right) - \bar{\mu}_0 \left( \bar{\eth} \mu_{1'} -\sqrt{2} \bar{\phi}_{0'1'} \bar{\lambda}_1+\rho'  \mu_{0'} \right) \right]
\eea
We now prove a set of simple Lemmas:

\begin{lemma}
\label{lem:vanishI}
If $S$ is a supersymmetric surface with spinor field $\epsilon$ then $\hat{I}_S[\epsilon]=0$.
\end{lemma}
\begin{proof}
If $S$ is supersymmetric then the pull-back to $S$ of the Nester-Witten $2$-form vanishes. 
\end{proof}

\begin{lemma}
\label{lem:hatI2}
If $\epsilon$ satisfies \eqref{hol_conds1} on $S$ then 
\be
\label{hatISfull}
 \hat{I}_S[\epsilon]= \int_S \left[ \rho' \left( |\lambda_{0'}|^2 + |\mu_{0'}|^2 \right) +  \rho \left( |\lambda_{1'}|^2 + |\mu_{1'}|^2 \right) \right] 
\ee
\end{lemma}
\begin{proof}
Following  \cite{Dougan:1991zz}, use  \eqref{hol_conds1} to eliminate the $\bar{\eth}$ terms from \eqref{hatIS}, then integrate by parts the two terms involving $\eth$ and use the complex conjugate of \eqref{hol_conds1} to eliminate the resulting ${\eth}$ terms. The terms involving $\bar{\phi}_{0'1'}$ end up making an imaginary contribution to the integral and hence drop out when we take the real part.
\end{proof}

\begin{lemma}
\label{lem:trapped}
A supersymmetric surface is not trapped.
\end{lemma}
\begin{proof}
A supersymmetric surface satisfies \eqref{hol_conds1} and so \eqref{hatISfull} holds. If the surface were trapped then it would have $\rho>0$ and $\rho'>0$. Since $\epsilon \ne 0$ this implies $\hat{I}_S[\epsilon]>0$, contradicting  Lemma \ref{lem:vanishI}.
\end{proof}

The next Lemma covers the case of a marginally trapped surface. Recall that $X^a$ and $V$ are defined by \eqref{XVdef}.
\begin{lemma}
\label{lem:marg}
If $S$ is a supersymmetric surface with $\rho'>0$ and $\rho\equiv 0$ then $V \equiv 0$ on $S$ and $X^a$ is null on $S$. 
\end{lemma}
\begin{proof}
From Lemmas \ref{lem:vanishI} and \ref{lem:hatI2} we must have $\lambda_{0'} = \mu_{0'}\equiv 0$ on $S$ which implies $V\equiv 0$ on $S$ so \eqref{Xnorm} implies that $X^a$ is null on $S$.
\end{proof}

The following result is a supercovariant version of a theorem of Dougan and Mason \cite{Dougan:1991zz}. A similar result has been obtained by Rogatko \cite{Rogatko:2002qe}. We will later apply this to the situation where $S$ is a cross-section of the event horizon of a black hole formed in gravitational collapse, with $\Sigma$ the black hole interior, the objective being to prove that this is impossible if $S$ is ``the same as'' a horizon cross-section of a supersymmetric black hole.

\begin{theorem}
\label{thm:Ipos}
Let $\Sigma$ be a smooth compact connected spacelike surface with smooth connected boundary $\partial \Sigma  =S$ in a smooth spacetime satisfying the local mass-charge inequality \eqref{local_mass_charge}. Assume that $\rho' > 0$ on $S$ and let $\epsilon$ be a non-zero solution of \eqref{hol_conds1} on $S$. Then $\hat{I}_S[\epsilon] \ge 0$ with equality if, and only if, (i) $S$ is a supersymmetric surface (with spinor $\epsilon$), (ii) $\epsilon$ extends to a spinor on  $\Sigma$ satisfying $h^a_b \hat{\nabla}_a \epsilon=0$ where $h^a_b$ is the projection onto $\Sigma$ and, (iii) on $\Sigma$, the charged matter satisfies
 \be
 \label{Sigma_matter}
  N^a \left( T^{(m)}_{ab} X^b - {\rm Re}(V) J_a + {\rm Im}(V)  \tilde{J}_a \right) = 0
 \ee
 where $N^a$ is a normal to $\Sigma$ and $X^a$, $V$ are defined by \eqref{XVdef}.
 \end{theorem}
\noindent (Note that if $n^a$ is chosen as the inward-pointing normal to $S$ then $\rho'>0$ is the statement that the inward directed null geodesics orthogonal to $S$ are converging.)
\begin{proof}
Following \cite{Dougan:1991zz}, let $\tilde{\epsilon}$ be any spinor field on $\Sigma$ satisfying 
\be
\label{wittenbcs}
 \tilde{\lambda}_{1'} = \lambda_{1'} \qquad \qquad \tilde{\mu}_{1'} = \mu_{1'} \qquad {\rm on} \; S
\ee
Consider $\hat{I}_S[\tilde{\epsilon}]$ written in the form \eqref{hatIS}. Following \cite{Dougan:1991zz} we can rearrange using \eqref{hol_conds1}, \eqref{wittenbcs} and integration by parts. Once again, the terms involving $\bar{\phi}_{0'1'}$ drop out when we take the real part, leaving
\be
\label{positivity}
 \hat{I}_S[\epsilon]  = \hat{I}_S[\tilde{\epsilon}] +  \int_S \rho' \left( |\tilde{\lambda}_{0'} - \lambda_{0'}|^2 +  |\tilde{\mu}_{0'} - \mu_{0'}|^2 \right)
\ee
Now assume that $\tilde{\epsilon}$ satisfies the supercovariant version \cite{Gibbons:1982jg} of the Sen-Witten equation on $\Sigma$ 
\be
\label{senwitten}
\gamma^b h^a_b \hat{\nabla}_a \tilde{\epsilon}=0
\ee
with boundary conditions \eqref{wittenbcs}. An explanation for why such a solution exists is given in Appendix \ref{app:spinor}. In \cite{Gibbons:1982fy,Gibbons:1982jg} it is shown how to convert $\hat{I}_S[\tilde{\epsilon}] \equiv \int_S \hat{\Lambda}[\tilde{\epsilon}] = \int_\Sigma d\hat{\Lambda}$ into a manifestly non-negative expression, provided the local mass-charge inequality is satisfied. Hence if $\rho' \ge 0$ on $S$ then the RHS of \eqref{positivity} is a sum of non-negative terms and so we have proved $\hat{I}_S[\epsilon] \ge 0$. 

If $\rho'>0$ on $S$ and $\hat{I}_S[\epsilon]=0$ then \eqref{positivity} and \eqref{wittenbcs} imply that $\tilde{\epsilon} = \epsilon$ on $S$ and so $\tilde{\epsilon}$ is an extension of $\epsilon$ so we can drop the tilde. 
Furthermore, from  \cite{Gibbons:1982fy,Gibbons:1982jg}, $\int_\Sigma d\hat{\Lambda}$ is a sum of two non-negative terms (see in particular equation A11 of \cite{Gibbons:1982jg}). The first vanishes if, and only if, $h^a_b \hat{\nabla}_a \epsilon=0$ on $\Sigma$ (which implies that $S$ is supersymmetric). The second vanishes if, and only if, \eqref{Sigma_matter} holds on $\Sigma$. Conversely, if (i), (ii) and (iii) hold then we can reverse the argument to deduce that $\hat{I}[\epsilon]=0$. 
\end{proof}

It follows from Theorem \ref{thm:Ipos} that if a supersymmetric surface has a compact interior then this interior does not contain a trapped surface:

\begin{corollary}
\label{cor:nottrapped}
If $\Sigma$ and $S$ satisfy the assumptions of Theorem \ref{thm:Ipos} and $S$ is a supersymmetric surface then 
every smooth 2d surface $S' \subset \Sigma$ is supersymmetric and hence (Lemma \ref{lem:trapped}) not trapped.
\end{corollary}
\begin{proof}
From the theorem there exists a not identically zero spinor $\epsilon$ satisfying $h^a_b \hat{\nabla}_a \epsilon=0$ on $\Sigma$. Arguing similarly to Lemma \eqref{lem:nonzero}, $\epsilon$ cannot vanish anywhere, in particular it does not vanish on $S'$. it follows that $S'$ is a supersymmetric surface.
\end{proof}

Although we shall not need to use it, we state the following result for completeness: 

\begin{lemma}
\label{lem:charged_dust}
If \eqref{Sigma_matter} holds and the local mass-charge inequality is satisfied then the following holds on $\Sigma$:
\be
\label{Sigma_matter2}
 T^{(m)}_{ab} = \chi X_a X_b \qquad {J}_a = \chi {\rm Re}(V) X_a \qquad {\tilde{J}}_a = - \chi {\rm Im}(V) X_a
 \ee
 where $\chi \ge 0$ is a scalar field on $\Sigma$. 
 \end{lemma}
\begin{proof}
See Appendix \ref{app:dust}. 
\end{proof}

At points where $V \ne 0$, \eqref{Sigma_matter2} describes charged dust. At points where $V =0$ it describes uncharged null dust. 

We can now present our main result. To motivate this, consider a cross-section $S$ of the event horizon of an extremal Reissner-Nordstr\"om black hole. This is a supersymmetric surface that is marginally trapped, i.e., it has $\rho'>0$ and $\rho \equiv 0$. Furthermore, $S$ has non-zero electric and/or magnetic charge. Now if such a black hole were formed in gravitational collapse then it would have compact interior, i.e., $S = \partial \Sigma$ where $\Sigma$ is a compact spacelike surface, and the total electric or magnetic charge on $\Sigma$ would be non-zero. The next theorem shows that this is impossible, i.e., the third law of black hole mechanics holds (we'll formulate this more precisely as a corollary below). Note that in a maximal analytic extension of the extremal Reissner-Nordstr\"om solution we {\it can} write $S = \partial\Sigma$ where $\Sigma$ is smooth and spacelike but such $\Sigma$ is {\it not compact} because of the curvature singularity at $r=0$, which is not part of the spacetime. 

\begin{theorem}
\label{thm:spinor}
Let $S$ be a smooth compact connected supersymmetric surface with $\rho'>0$ and $\rho \equiv 0$ in a smooth spacetime satisfying the local mass-charge inequality. If there exists a smooth compact connected spacelike surface $\Sigma$ with boundary $\partial \Sigma  =S$ then the total electric and magnetic charges on $\Sigma$ must vanish. 
\end{theorem}
\begin{proof}
We shall refer to our spacetime $M$ as the physical spacetime. From Theorem \ref{thm:Ipos} the spinor $\epsilon$ on $S$ can be extended to a spinor field satisfying $h^a_b \hat{\nabla}_a \epsilon=0$ on $\Sigma$. In Appendix B of \cite{Chrusciel:2005ve} it is shown that one can construct a spacetime $\hat{M}$, containing a Maxwell field, the {\it Killing development} of the initial data on ${\rm int}(\Sigma)$ with the following properties. (i) $\hat{M}$ contains a surface diffeomorphic to $\Sigma$, which we shall also denote as $\Sigma$; (ii) the induced metric, extrinsic curvature, and Maxwell field on ${\rm int}(\Sigma)$ in $\hat{M}$ agree with those in the physical spacetime; (iii) $\epsilon$ extends to a supercovariantly constant spinor on $\hat{M}$. Using (iii) we extend the definitions \eqref{XVdef} of $X^a$ and $V$ to $\hat{M}$ and (iii) implies that $X^a$ is a Killing vector field that preserves the Maxwell field \cite{Tod:1983pm,Chrusciel:2005ve}. 

At this stage $\hat{M}$ is just an auxiliary, unphysical, spacetime (but see the remark following this proof). From the Einstein tensor we can define the energy-momentum tensor $\hat{T}_{ab}$ of $\hat{M}$ which we split into a Maxwell part $\hat{T}_{ab}^{(Max)}$ defined in the usual way and a ``matter'' part $\hat{T}_{ab}^{(m)}$. Similarly from the derivative of the Maxwell field we can define electric and magnetic currents $\hat{J}_a$ and $\hat{\tilde{J}}_a$ on $\hat{M}$. The existence of a supercovariantly constant spinor constrains these quantities. If $X^a$ is timelike at a point $p \in \hat{M}$ then $X^a$ is timelike in a neighbourhood of $p$. In this neighbourhood, the results of Tod \cite{Tod:1983pm} show that (compare Lemma \ref{lem:charged_dust})
\be
\label{emJKillingdev}
 \hat{T}^{(m)}_{ab} = \hat{\chi} X_a X_b \qquad \hat{J}_a = \hat{\chi} {\rm Re}(V) X_a \qquad \hat{\tilde{J}}_a = - \hat{\chi} {\rm Im}(V) X_a
 \ee
 for some function $\hat{\chi}$ invariant under the flow of $X^a$. If $X^a$ is null at $p$ and throughout a neighbourhood of $p$ the the above results also hold in this neighbourhood \cite{Tod:1983pm}. The remaining possibility is that $X^a$ is null at $p$ but every neighbourhood of $p$ contains points at which $X^a$ is timelike. In this case there exists a sequence of points $p_n \rightarrow p$ such that $X^a$ is timelike at each $p_n$. Then $X^a$ is timelike in a neighbourhood of each $p_n$ and so \eqref{emJKillingdev} holds at $p_n$. By continuity it also holds at $p$. This proves that \eqref{emJKillingdev} holds throughout $\hat{M}$. 

Existence of a supercovariantly constant spinor does not constrain the sign of $\hat{\chi}$. To do this we need to relate these results to the physical spacetime, in which the local mass-charge inequality is satisfied. Using (ii), the Hamiltonian constraint on $\Sigma$ implies that $\hat{T}_{ab} N^a N^b = T_{ab} N^a N^b$ on $\Sigma$ (which has unit normal $N^a$) and hence $\hat{T}^{(m)}_{ab} N^a N^b = {T}^{(m)}_{ab} N^a N^b$ on $\Sigma$ (since the Maxwell field on $\Sigma$ is the same in $\hat{M}$ and  in the physical spacetime). The weak energy condition (first equation of \eqref{local_mass_charge}) implies that the RHS is non-negative which implies $\hat{\chi} \ge 0$.

Now we impose the condition that $\rho \equiv 0$ on $S$. From Lemma \ref{lem:marg} we know that $V \equiv 0$ on $S$ and $X^a$ is null on $S$. Now in $\hat{M}$ we have the equation \cite{Tod:1983pm,Chrusciel:2005ve} 
\be
\label{nabX}
 \nabla_a X_b = \bar{V} \phi_{AB}\epsilon_{A'B'} + V \bar{\phi}_{A'B'} \epsilon_{AB} 
\ee
Since $V$ vanishes on $S$ it appears that $\nabla_a X_b$ vanishes on $S$. However, this is a little delicate since $S$ is at the boundary of $\Sigma$ and we haven't shown that $X_a$ is differentiable in directions transverse to $S$. We could work with $1$-sided derivatives but instead we can proceed as follows. Consider a set of smooth $2$-surfaces $S_\delta \subset {\rm int} (\Sigma)$ such that $S_\delta \rightarrow S$ as $\delta \rightarrow 0$. Let $\Sigma_\delta$ be the region of $\Sigma$ enclosed by $S_\delta$ so $\Sigma_\delta \rightarrow \Sigma$ as $\delta \rightarrow 0$.  \eqref{nabX} implies $dX|_{S_\delta} \rightarrow 0$ as $\delta \rightarrow 0$. Hence we have the Komar-like identity following from the fact that $X^a$ is a Killing field on $\hat{M}$:
\be
 0 = \lim_{\delta \rightarrow 0} \int_{S_\delta} \star dX =  \lim_{\delta \rightarrow 0} \int_{\Sigma_\delta} d \star dX =  \pm 2 \lim_{\delta \rightarrow 0} \int_{\Sigma_\delta} \hat{R}_{ab} N^a X^b = \pm 2 \int_\Sigma \hat{R}_{ab} N^a X^b
\ee
where $\hat{R}_{ab}$ is the Ricci tensor on $\hat{M}$ and $\pm$ is because we haven't been careful to keep track of signs here. Hence
\bea
 0 &=& \int_\Sigma \hat{R}_{ab} N^a X^b = 8 \pi \int_\Sigma \left[ T_{ab}^{(Max)} N^a X^b + \left( \hat{T}_{ab}^{(m)} - \frac{1}{2} \hat{T}^{c(m)}_c \hat{g}_{ab} \right) N^a X^b \right] \nonumber \\ &=& 8\pi  \int_\Sigma \left[ T_{ab}^{(Max)} N^a X^b + \frac{1}{2} \hat{\chi} X^2 N_a X^a \right] 
\eea
where we used the fact (ii) that the Maxwell field on $\Sigma$ is the same in $\hat{M}$ and in the physical spacetime, and
$T^{a(Max)}_a = 0$. Both terms in the integrand are non-negative (the Maxwell field satisfies the dominant energy condition) and so must both vanish. In particular on $\Sigma$ we have $0 = \hat{\chi}X^2 = \hat{\chi}V\bar{V}$ and hence $\hat{\chi} V= 0$ so $\hat{J}_a = \hat{\tilde{J}}_a = 0$ on $\Sigma$. In particular, the electric and magnetic charge densities on $\Sigma$ vanish in $\hat{M}$ (and in fact throughout $\hat{M}$ by the Killing symmetry). Hence $\int_{S_\delta} F = \int_{S_\delta} \star F = 0$ holds in $\hat{M}$. These results also hold in the physical spacetime because the Maxwell field on $\Sigma$ is the same in both spacetimes. Taking the limit $\delta \rightarrow 0$ shows that the electromagnetic charges of $S$ vanish, so the total electric and magnetic charges on $\Sigma$ must vanish.\footnote{We note that combining the Hamiltonian constraint with Lemma \ref{lem:charged_dust} gives $\hat{\chi} (N\cdot X)^2 = \chi (N \cdot X)^2$ so $\hat{\chi} = \chi$. Hence $\chi V \bar{V}=0$ and so from Lemma \ref{lem:charged_dust} we also have $J^a = \tilde{J}^a=0$ on $\Sigma$. This is non-trivial because the spatial components of the currents are {\it not} determined by constraint equations on $\Sigma$ and so {\it a priori} might differ in the physical and unphysical spacetimes.}

\end{proof}

Note that we could replace the condition $\rho \equiv 0$ in the statement of this theorem with the weaker condition that $X^a$ is null on $S$ (Lemma \ref{lem:marg} shows that the former implies the latter but they are not equivalent). Even  if we drop the assumption that $\rho \equiv 0$ or that $X^a$ is null on $S$ then we expect that, within $D(\Sigma)$ (the domain of dependence of $\Sigma$), the metric and Maxwell field will be diffeomorphic to those in the corresponding region $\hat{D}(\Sigma)$ of the Killing development $\hat{M}$ discussed above, and so there will exist a supercovariantly constant spinor field throughout $D(\Sigma)$. However, we shall not attempt to give a proof of this result (see \cite{szabados:1993} for the case of vanishing Maxwell field).

As an example of a surface satisfying the conclusions of this theorem, consider a spacetime admitting a supercovariantly constant spinor for which $X^a$ is globally null. Tod's analysis \cite{Tod:1983pm} shows that such a spacetime is a pp-wave with vanishing electric and magnetic currents. Any compact spacelike $2$-surface in such a spacetime is a supersymmetric surface with vanishing electric and magnetic charges. In fact, it seems likely that if the conditions of our theorem are satisfied (replacing $\rho \equiv 0$ with $X^a$ null) then the spacetime in $D(\Sigma)$ must always be such a pp-wave (or flat), as in \cite{szabados:1993}.

We can now state a version of the third law of black hole mechanics for supersymmetric black holes:

\begin{corollary}[Third law for supersymmetric black holes.]
\label{cor:3rdlaw}
let $S$ be a smooth, compact, connected, 2-surface in a smooth spacetime $M$ with $S = \partial \Sigma$ where $\Sigma$ is a smooth, compact, connected, spacelike hypersurface. Let $\tilde{M}$ be an asymptotically flat black hole spacetime admitting a supercovariantly constant spinor, for which the future event horizon ${\cal H}^+$ is a Killing horizon. Assume that there exists a diffeomorphism $\Phi$ from a neighbourhood  $U$ of $S$ in $M$ to a neighbourhood of a cross-section $\tilde{S}$ of ${\cal H}^+$ in $\tilde{M}$ such that $\Phi(S) = \tilde{S}$, $\Phi$ maps $\Sigma\cap U$ to the black hole interior, and $\Phi$ maps the (spacetime) metric, extrinsic curvature, and Maxwell field on $S$ to the corresponding quantities on $\tilde{S}$. Assume that the ingoing null geodesics normal to $\tilde{S}$ are strictly converging. If the matter satisfies the local mass-charge inequality then the electric and magnetic charges of $S$ and $\tilde{S}$ are zero. In particular, $\tilde{S}$ cannot be a cross-section of an extremal Reissner-Nordstr\"om horizon. 
\end{corollary}
\noindent 
A few comments before we prove this result. First, the existence of a supercovariantly constant spinor in $\tilde{M}$ implies that $X^a$ is a causal Killing vector field. If this Killing field has complete orbits then it generates a symmetry of the spacetime, and so must preserve ${\cal H}^+$, which implies that it is tangent to the generators of ${\cal H}^+$, so ${\cal H}^+$ is indeed a Killing horizon. In particular this assumption is true if the black hole spacetime is a maximal analytic extension of extremal Reissner-Nordstr\"om. Similarly the assumption about ingoing null geodesics being converging is true for extremal Reissner-Nordstr\"om. However, the result also covers a larger class of supersymmetric black hole spacetimes, e.g., extremal Reissner-Nordstr\"om in equilibrium with supersymmetric charged matter outside the horizon. 

Second, the assumed properties of $\Phi$ capture the idea that $S$ is ``the same'' as a horizon cross-section of a supersymmetric black hole. We are considering the possibility that such a black hole forms in gravitational collapse, with the interior at a particular instant of time corresponding to $\Sigma$. Note that we assume that $\Phi$ maps the full spacetime metric on $S$ to that on $\tilde{S}$, not just the induced metric. 

Third, if we drop the assumption that the local mass-charge inequality is satisfied then the corresponding statement is {\it not} true: the spacetimes constructed by Kehle and Unger describe gravitational collapse to form a black hole that is exactly extremal Reissner-Nordstr\"om after a finite advanced time (Fig. \ref{fig:KU}) and so, after this time, any cross-section $S$ of the horizon would satisfy the assumptions of our theorem.

\begin{proof}
Given a null tetrad on $\tilde{S}$ we can use $\Phi$ to pull it back to define a null tetrad on $S$. The assumed properties of $\Phi$ imply that the ``inward'' null normal of $\tilde{S}$ pulls back to the inward null normal $n^a$ on $S$. It now follows that $\Phi$ maps the derivative $\eth$ on $S$ to the corresponding quantity on $\tilde{S}$. This is because the ``connection terms'' in the definition of $\eth$ that are not determined by the intrinsic geometry of $S$ are $n^b m^a \nabla_a \ell_b$ and $n^b \bar{m}^a \nabla_a \ell_b$, which depend only on tangential derivatives of $\ell^a$ and $n^a$.  The assumption that $\Phi$ maps the extrinsic curvature and Maxwell field on $S$ to the corresponding quantities on $\tilde{S}$ implies that, on $S$, the quantities $\rho,\rho',\sigma,\sigma'$ and $\phi_{AB}$ are pull-backs of the corresponding quantities on $\tilde{S}$. In particular on $S$ we have $\rho'>0$ (because the ingoing null geodesics normal to $\tilde{S}$ are strictly converging) and $\rho \equiv 0$ (because ${\cal H}^+$ is a Killing horizon, so its generators have vanishing expansion and shear). 

It now follows that equations \eqref{hol_conds1}, \eqref{hol_conds2}, \eqref{anti_hol_conds1}, \eqref{anti_hol_conds2} take exactly the same form on $S$ as on $\tilde{S}$. On $\tilde{S}$ we have a non-trivial solution $\tilde{\lambda}_{A'}$, $\tilde{\mu}_{A'}$ of these equations so we must have a non-trivial solution $\lambda_{A'}$, $\mu_{A'}$ of these equations on $S$. (This solution is a pull-back of the solution on $\tilde{S}$.) Hence $S$ is a supersymmetric surface.  Theorem \eqref{thm:spinor} now implies that the electromagnetic charges of $S$ vanish.

\end{proof}

One can construct very artificial examples of supersymmetric black holes that evade this version of the third law as follows. Any spacetime belonging to the Majumdar-Papapetrou family admits a supercovariantly constant spinor with timelike $X^a$ \cite{Tod:1983pm}. This family contains multi-extremal-Reissner-Nordstr\"om solutions but it also contains non-black-hole solutions describing static configurations of supersymmetric charged matter with gravitational and electromagnetic forces cancelling (see e.g. \cite{bonner:1975}). In such a spacetime, delete a spacelike disc of points. This creates a non-trivial black hole region to the past of the deleted set. This black hole has compact interior and non-zero charge. It evades our third law because ${\cal H}^+$ is not a Killing horizon, and the horizon generators are expanding, i.e.,  $\rho<0$. This is possible because $X^a$ does not have complete orbits, i.e., it does not generate a symmetry of the spacetime and therefore does not have to preserve ${\cal H}^+$. 
This example is unphysical because it is not a Cauchy development of initial data. If we pick a surface in this spacetime to the past of the black hole region then the maximal Cauchy development of the data on this surface is not the black hole spacetime, it is the original spacetime, with the deleted points reinstated.

There is an alternative version of the third law that holds even for this type of black hole. To see this, make the same asumptions as in Corollary \ref{cor:3rdlaw} but do not require ${\cal H}^+$ to be a Killing horizon. Then $S$ is a supersymmetric surface (as explained in the proof of Corollary \ref{cor:3rdlaw}) so Corollary \ref{cor:nottrapped} shows that $\Sigma$ cannot contain a trapped surface. This proves that a black hole with a compact interior containing a trapped surface cannot ``become supersymmetric'' in any theory with matter satisfying the local mass-charge inequality. This is a version of the third law in the spirit of \cite{Israel:1986gqz}. The work of Kehle and Unger proves that the corresponding statement is false if the matter is a massless charged scalar field \cite{Kehle:2022uvc}. 

\section{A quasi-local mass-charge inequality}

\label{sec:quasilocal}

Given a smooth spacelike compact $2$-surface $S$, Dougan and Mason showed how to define a quasi-local $4$-momentum $P^a$ and mass $M$ for $S$ that satisfies a number of nice properties \cite{Dougan:1991zz}. If $S = \partial \Sigma$ with $\Sigma$ a compact spacelike surface and matter satisfies the dominant energy condition then $P^a$ is non-spacelike (so $M \ge 0$) and vanishes if, and only if, $D(\Sigma)$ is flat \cite{szabados:1993}.

In this section we will use a supercovariant version of Dougan and Mason's approach to define a ``renormalized'' quasi-local mass $\varpi$, differing from $M$ by a term involving the electromagnetic field. Unlike $M$, $\varpi$ is equal to the ADM mass when evaluated for a symmetry $2$-sphere in the Reissner-Nordstr\"om solution. If $S= \partial \Sigma$ with $\Sigma$ a compact spacelike surface and matter satisfies the local mass-charge inequality \eqref{local_mass_charge} then we will show that $\varpi^2 \ge Q^2 + P^2$, where $Q$ and $P$ are the electric and magnetic charges of matter on $\Sigma$,  with equality only if there exists a supercovariantly constant spinor on $\Sigma$.

\subsection{The Dougan-Mason quasilocal mass}

We start by recalling the definition of the Dougan-Mason (DM) $4$-momentum and mass \cite{Dougan:1991zz}. Let $S$ be a smooth connected compact spacelike $2$-surface (e.g. a 2-sphere). Choose a spinor basis $(o^A,i^A)$ such that $\ell^a=o^A \bar{o}^{\u{A}'}$ and $n^a=i^A \bar{i}^{\u{A}'}$ are the ``outward'' and ``inward'' null normals to $S$. On $S$ consider ``holomorphic'' 2-component spinors, i.e., those satisfying the following equation on $S$:
\be
\label{hol_eq}
\bar{m}^a \nabla_a \lambda_{B'} = 0
\ee
This is equivalent to
\begin{subequations}
\be
\label{DM1}
  \bar{\eth} \lambda_{1'}  + \rho' \lambda_{0'}= 0 
\ee
\be
 \label{DM2}
\bar{\eth} \lambda_{0'} + \bar{\sigma} \lambda_{1'} = 0 
\ee
\end{subequations}
DM argue that, generically, \eqref{hol_eq} admits a $2$-dimensional space of solutions and two linearly independent solutions are also pointwise linearly independent (as 2-component spinors). However, there are ``exceptional'' situations where the solution space might be larger or the two solutions might fail to be pointwise linearly independent. Following DM, we assume that $S$ is non-exceptional.\footnote{An example of an exceptional case is a sphere on the white hole horizon of the Schwarzschild solution, which has $\rho'=\sigma=0$ (assuming $n^a$ is tangent to the white hole horizon generators). In this case $\lambda_{1'}$ is an arbitrary linear combination of spin-weighted spherical harmonics (see \eqref{spin_harm} below) ${}_{1/2} Y_{1/2 \pm 1/2}$ and $\lambda_{0'}$ is an arbitrary linear combination of ${}_{-1/2} Y_{1/2 \pm 1/2}$, so there is a 4d space of solutions. However, one can instead work with the DM definition using anti-holomorphic spinors, defined by relacing $\bar{m}^a$ with $m^a$ in \eqref{hol_eq}. Both definitions fail at the bifurcation surface \cite{Szabados:2009eka}.}

 Let $\lambda_{A'}^{\underline{A}'}$ be two linearly independent solutions of \eqref{hol_eq} where $\underline{A}' \in \{0',1'\}$ labels the solutions. Define a ``quasitranslation'' to be a 4-vector field on $S$ of the form $K_{AA'} = K_{\underline{AA}'} \lambda_{A'}^{\underline{A}'} \bar{\lambda}^{\underline{A}}_A$ where $K_{\u{AA}'}$ are constants. The energy associated with $K_{AA'}$ is then defined as $P^{\underline{AA}'}  K_{\underline{AA}'}$ where the {\it Dougan-Mason 4-momentum} is\footnote{One can verify that this expression is real \cite{Penrose:1986ca}.}
\be
 P^{\underline{AA}'} = \frac{i}{4\pi} \int_S \bar{\lambda}_A^{\u{A}} \nabla_{BB'} \lambda_{A'}^{\u{A}'} \, \theta^{BB'} \wedge \theta^{AA'}
\ee
where $\theta^{AA'}$ are a basis of $1$-forms. Now define
\be
 \epsilon^{\underline{A}' \underline{B}'} = \epsilon^{A'B'} \lambda_{A'}^{\underline{A}'} \lambda_{B'}^{\underline{B}'}
\ee
which is constant as a consequence of \eqref{hol_eq} and Liouville's theorem. The assumption that  $\lambda_{A'}^{\underline{0}'}$ and $\lambda_{A'}^{\underline{1}'}$ are pointwise linearly independent implies that $\epsilon^{\underline{A}' \underline{B}'}$ is non-zero. Define $\epsilon_{\underline{A}' \underline{B}'}$ to be the inverse of  $\epsilon^{\underline{A}' \underline{B}'}$, i.e., $\epsilon_{\underline{A}' \underline{B}'}  \epsilon^{\underline{B}' \underline{C}'} = -\delta^{\underline{C}'}{}_{\underline{A}'}$. The {\it Dougan-Mason quasilocal mass} is then defined by
\be
\label{mdef}
 M^2 \equiv  P^{\underline{AA}'}  P^{\underline{BB}'}\bar{ \epsilon}_{\underline{A}\underline{B}}  \epsilon_{\underline{A}'\underline{B}'} 
\ee

\begin{theorem}[Dougan-Mason  \cite{Dougan:1991zz}]
\label{thm:DM}
 Assume that $S = \partial \Sigma$ where $\Sigma$ is a smooth compact connected spacelike surface, that the dominant energy condition is satisfies on $\Sigma$, and that $\rho' \ge 0$ on $S$. Then $P^{\underline{AA}'}$ is a future-directed causal vector (so $M \ge 0$ is well-defined by \eqref{mdef}).
\end{theorem}
\begin{proof}
(The proof of theorem \ref{thm:Ipos} is essentially a ``supercovariant version'' of the first part of the following proof so we've seen most of the following steps already.) Let $\lambda_{A'}$ be a solution of \eqref{DM1} and introduce the Nester-Witten \cite{Witten:1981mf,Nester:1981bjx} 2-form (compare \eqref{hatLambda}, here we don't need to take the real part)
\be
 \Lambda_{AA'BB'} =-i\left( \bar{\lambda}_A \nabla_{BB'} \lambda_{A'} - i \bar{\lambda}_B \nabla_{AA'} \lambda_{B'} \right)
\ee
Now define the (real) functional 
\be
\label{ISlambda_def}
 I_S[\lambda]  \equiv \int_S \Lambda 
 \ee
Written out in GHP notation we have (compare \eqref{hatIS})
\be 
\label{ISlambda}
I_S[\lambda]=  \int_S \left[\bar{\lambda}_1 \left( \eth \lambda_{0'} + \rho \lambda_{1'} \right) - \bar{\lambda}_0 \left( \bar{\eth} \lambda_{1'} + \rho' \lambda_{0'} \right) \right]
\ee
Using \eqref{DM1} in \eqref{ISlambda}, integrating by parts, and using the complex conjugate of \eqref{DM1} gives (compare \eqref{hatISfull})
\be
\label{ISrho}
 I_S[\lambda] = \int_S\left( \rho'|\lambda_{0'}|^2 + \rho |\lambda_{1'}|^2 \right)
\ee
This is not obviously non-negative because $\rho$ might be negative (this is the usual situation). We now consider a solution $\tilde{\lambda}_{A'}$ of the Sen-Witten equation on $\Sigma$ with $\tilde{\lambda}_{1'}=\lambda_{1'}$ on $S$ to obtain $I_S[\lambda] = I_S[\tilde{\lambda}] + \int_S \rho' |\tilde{\lambda}_{0'} - \lambda_{0'}|^2$. Converting $I_S[\tilde{\lambda}]$ to an integal over $\Sigma$ and using the Sen-Witten equation and dominant energy condition gives $I_S[\tilde{\lambda}] \ge 0$. Using $\rho' \ge 0$ now implies $I_S[\lambda] \ge 0$.

We now assume $\lambda_{A'}$ satisfies \eqref{DM2} as well as \eqref{DM1} so we can write it in terms of our basis of solutions of as $\lambda_{A'} = \lambda_{\underline{A}'} \lambda_{A'}^{\underline{A}'}$ for some constants $\lambda_{\underline{A}'}$. We then have
\be
\label{PIS}
\frac{1}{4\pi} I_S[\lambda]=  P^{\underline{AA}'} \bar{\lambda}_{\underline{A}} \lambda_{\underline{A}'} 
\ee
It follows that the contraction of $P^{\underline{AA}'}$ with any future-directed null vector $\bar{\lambda}_{\underline{A}} \lambda_{\underline{A}'}$  is non-negative, so $P^{\underline{AA}'}$ is future-directed causal (or zero). 
\end{proof}

Equation \eqref{DM1} was used repeatedly in the above proof. However, the only role of \eqref{DM2} is to help define the 2d space of functions labelled by $\u{A}'$. Similarly we did not use \eqref{hol_conds2} in the proof of Theorem \ref{thm:Ipos}. It has been observed previously \cite{bergqvist:1992} that \eqref{DM2} can be modified (e.g. multiplying the final term by a suitable function) without affecting the positivity properties of the DM 4-momentum and mass or its agreement with other definitions (Bondi, ADM, Hawking) in the appropriate situations.

\subsection{Renormalized Dougan-Mason 4-momentum and mass}

We start by reformulating the definition of the DM 4-momentum. Assume $\rho'>0$ on $S$ (i.e. the ingoing null geodesics normal to $S$ are strictly converging) and use \eqref{DM1} to eliminate $\lambda_{0'}$ from \eqref{ISrho}:
\be
   I_S[\lambda] = \int_S\left( \frac{1}{\rho'} |\bar{\eth}\lambda _{1'}|^2 + \rho |\lambda_{1'}|^2 \right)
\ee
The proof of theorem \ref{thm:DM} shows that this expression is non-negative for {\it any} $\lambda_{1'}$ (as just remarked this part of the proof did not use \eqref{DM2}).\footnote{Equivalently, if we define a scalar operator $O$ on $S$ by
$
 O f = -\rho' \eth \left( {\rho'}^{-1} \bar{\eth} f \right) + \rho \rho' f 
$
then the above DM argument implies that $O$ is a non-negative operator when acting on quantities with the same boost and spin weights as $\lambda_{1'}$.} We can also write the 4-momentum entirely using $\lambda_{1'}$ by eliminating $\lambda^{\underline{A}'}_{0'}$ to obtain
\be
\label{P1}
 P^{\underline{AA'}} = \frac{1}{4\pi} \int_S \left( \frac{1}{\rho'} \eth \bar{\lambda}_1^{\underline{A}}  \bar{\eth}  \lambda^{\underline{A}'}_{1'} + \rho \bar{\lambda}_1^{\underline{A}} \lambda^{\underline{A}'}_{1'} \right)
\ee
Here $\lambda_{1'}^{\underline{A}'}$ can be defined as a pair of linearly independent solutions of the equation obtained by combining \eqref{DM1} and \eqref{DM2}:
\be
\label{DM2a}
 - \bar{\eth}\left( \frac{\bar{\eth} \lambda_{1'}}{\rho'} \right) + \bar{\sigma} \lambda_{1'} = 0
\ee
This equation has a 2d space of solutions for non-exceptional $S$ since we know that \eqref{DM1} and \eqref{DM2} have a 2d space of solutions. We also have
\be
\label{eps1}
 \epsilon^{\underline{A}' \underline{B}'} = -\frac{1}{\rho'} \left( \bar{\eth} \lambda^{\underline{A}'}_{1'} \lambda^{ \underline{B}'}_{1'} - \bar{\eth} \lambda^{\underline{B}'}_{1'} \lambda^{ \underline{A}'}_{1'} \right)
\ee
which is constant by \eqref{DM2a} and Liouville's theorem, and non-zero by the assumption that $S$ is non-exceptional.

Our renormalized 4-momentum is defined for a (non-exceptional) surface $S$ with $\rho'>0$ by
\be
\label{Pren}
 P_{\rm ren}^{\u{AA}'} = P^{\u{AA}'} + \frac{1}{4\pi} \int_S \frac{2}{\rho'} |\phi_{01}|^2 \bar{\lambda}_1^{\underline{A}} \lambda^{\underline{A}'}_{1'} \
\ee
where $\lambda_{1'}^{\underline{A}'}$ are a pair of solutions of \eqref{DM2a} and $P^{\u{AA}'}$ is the DM 4-momentum given by \eqref{P1}. We define the renormalized DM mass $\varpi$ by
\be
 \varpi^2 \equiv P_{\rm ren}^{\underline{AA}'}  P_{\rm ren}^{\underline{BB}'} \bar{\epsilon}_{\underline{A}\underline{B}}  \epsilon_{\underline{A}'\underline{B}'} 
\ee
with $\epsilon_{\underline{A}' \underline{B}'}$ the inverse of (constant) $\epsilon^{\underline{A}' \underline{B}'}$ defined by \eqref{eps1}. If $\varpi^2 \ge 0$ then we choose $\varpi \ge 0$. 

The difference between $P_{\rm ren}^{\u{AA}'}$ and $P^{\u{AA}'}$ is quadratic in the Maxwell field hence to linear order in perturbations around Minkowski spacetime (treating the Maxwell field as a first order quantity) we have $P_{\rm ren}^{\u{AA}'} = P^{\u{AA}'}$ and $\varpi=M$. Hence $P_{\rm ren}^{\u{AA}'}$ and $\varpi$ inherit the properties of the DM 4-momentum and mass at zeroth order and at linear order. At zeroth order this implies that $P_{\rm ren}^{\u{AA}'}$ and $\varpi$ vanish for a surface in Minkowski spacetime. At linear order, $P_{\rm ren}^{\u{AA}'}$ and $\varpi$ are equal to the 4-momentum and mass of the matter enclosed by $S$, if the matter is viewed as a linear source \cite{Dougan:1991zz}.  

As an example, consider a spherically symmetric spacetime, taking $S$ to be a symmetry $2$-sphere of area-radius $r$. The DM mass for such a surface was calculated in \cite{dougan:1992}. Spherical symmetry implies that $\sigma=0$ and $\rho'$, $\rho$, $\phi_{01}$ are constant on $S$. We can write the definition of the charges \eqref{QPdef} as
\be
\label{QPdefsym}
 Q-iP = \frac{1}{2\pi} \int_S \phi_{01}
\ee
and hence in spherical symmetry we have $\phi_{01} = (Q-iP)/(2r^2)$. We shall use spin-weighted spherical harmonics ${}_{s} Y_{jm}$, see \cite[section 4.15]{Penrose:1985bww}. These harmonics obey \cite[eq 4.15.106]{Penrose:1985bww}
\be
\label{spin_harm}
 \bar{\eth} {}_{s} Y_{jm} = \left[ \frac{(j-s+1)(j+s)}{2r^2} \right]^{1/2} {}_{s-1} Y_{jm}
\ee
$\lambda_{1'}$ has spin-weight $1/2$ so it can be expanded in terms of ${}_{1/2} Y_{jm}$. Equation \eqref{DM2a} reduces to $\bar{\eth}\bar{\eth} \lambda_{1'}=0$. But $\bar{\eth}\bar{\eth} {}_{1/2} Y_{jm}$ vanishes if, and only if, $j=1/2$ (and $m = \pm 1/2$). So two linearly independent solutions of \eqref{DM2a} are $\lambda_{1'}^{\u{1}'} \equiv {}_{1/2} Y_{1/2 \, 1/2}$ and $\lambda_{1'}^{\u{0}'} \equiv {}_{1/2} Y_{1/2 \, -1/2}$. This choice gives
\be
 P_{\rm ren}^{\u{01}'} = P_{\rm ren}^{\u{10}'}=0 \qquad P_{\rm ren}^{\u{00}'} = P_{\rm ren}^{\u{11}'}  =\frac{1}{4\pi \rho'}  \left( \frac{1}{2} + r^2 \rho \rho' + \frac{Q^2+P^2}{2r^2}\right) =\frac{1}{4\pi \rho'}  \left( \frac{m}{r} + \frac{Q^2+P^2}{2r^2} \right)
\ee
where $m$ is the Hawking mass (see e.g. \cite{dougan:1992}). To compute our quasilocal mass $\varpi$ we should first compute $\epsilon^{\underline{A}' \underline{B}'}$ but we can avoid this work by noting that the above expression differs from the standard DM 4-momentum just by the $Q^2+P^2$ term above. Our renormalized DM mass is related to the renormalized DM 4-momentum in exactly the same way as the DM mass is related to the DM 4-momentum. The DM mass in spherical symmetry is just $m$ \cite{dougan:1992} and so it follows from our expression above that the renormalized DM mass coincides with the ``renormalized Hawking mass'':
\be
 \varpi = m + \frac{Q^2+P^2}{2r}
\ee 
For a symmetry sphere in a Reissner-Nordstr\"om spacetime this implies that $\varpi$ equals the ADM mass. 

We can now present our main result in this section:
\begin{theorem}
\label{thm:masscharge}
 Let $\Sigma$ be a smooth, compact, connected, spacelike $3$-surface with compact connected boundary $S = \partial \Sigma$. Assume $\rho'>0$ on $S$ and that the matter fields on $\Sigma$ satisfy the local mass charge inequality \eqref{local_mass_charge} on $\Sigma$. Then $P_{\rm ren}^{\u{AA}'}$ is non-spacelike and $\varpi$ satisfies a quasi-local mass-charge inequality:
\be
\label{masscharge}
 \varpi^2 \ge Q^2 + P^2 
\ee 
If this inequality is saturated then the conclusions (i),(ii),(iii) of theorem \ref{thm:Ipos} hold.
\end{theorem}
\begin{proof}
Recall the definition of $\hat{I}_S[\epsilon]$, equation \eqref{hatIS}. Assume that $\epsilon$ satisfies \eqref{hol_conds1}.
From theorem \ref{thm:Ipos} we have $\hat{I}_S[\epsilon] \ge 0$. Using $\eqref{hol_conds1}$ to eliminate $\lambda_{0'}$ and $\mu_{0'}$ gives
\be
   \hat{I}_S[\epsilon] = \int_S \left[\frac{1}{\rho'} \left( |\bar{\eth} \lambda_{1'} + \sqrt{2} \bar{\phi}_{0'1'} \bar{\mu}_1 |^2 +  |\bar{\eth} \mu_{1'} - \sqrt{2} \bar{\phi}_{0'1'} \bar{\lambda}_1 |^2\right) +  \rho \left( |\lambda_{1'}|^2 + |\mu_{1'}|^2 \right) \right] 
\ee
Expanding this out gives
\bea
    \hat{I}_S[\epsilon] &=& \int_S \left[ \frac{1}{\rho'} \left( |\bar{\eth} \lambda_{1'}|^2 + |\bar{\eth} \mu_{1'}|^2 \right) + \left( \rho + \frac{2}{\rho'} |\phi_{01}|^2 \right) \left( |\lambda_{1'}|^2 + |\mu_{1'}|^2 \right) \right]  \nonumber \\
    &+&  {\rm Re} \int_S \frac{2 \sqrt{2}}{\rho'} \phi_{01} \left( \bar{\eth} \lambda_{1'} \mu_{1'} - \bar{\eth} \mu_{1'} \lambda_{1'} \right)
\eea
We now assume that $\lambda_{1'}$ and $\mu_{1'}$ both satisfy \eqref{DM2a}.\footnote{
Note that we do {\it not} replace derivatives with supercovariant derivatives in \eqref{DM2a}. See the remark following this proof.
} Recall that $\lambda^{\u{A'}}_{1'}$ is a basis of solutions for \eqref{DM2a} so we can expand
\be
 \lambda_{1'} = \lambda_{\u{A}'} \lambda^{\u{A'}}_{1'} \qquad \qquad \mu_{1'} = \mu_{\u{A}'} \lambda^{\u{A'}}_{1'} 
\ee 
for constants $\lambda_{\u{A}'}$ and $\mu_{\u{A}'}$. Recall the definition of the renormalized 4-momentum in equation \eqref{Pren}, the constants $\epsilon^{\u{A}'\u{B}'}$ in \eqref{eps1} and the electromagnetic charges in \eqref{QPdefsym}. We obtain
\be
\label{ISQP}
  \frac{1}{4\pi} \hat{I}_S[\epsilon] = P_{\rm ren}^{\u{AA}'} \left( \bar{\lambda}_{\u{A}} \lambda_{\u{A}'} + \bar{\mu}_{\u{A}} \mu_{\u{A}'} \right) + \sqrt 2 {\rm Re} \left(  \epsilon^{\u{A}'\u{B}'} \lambda_{\u{A}'} \mu_{\u{B}'}(Q-iP) \right)
\ee
We know that the LHS is non-negative for all $\lambda_{\u{A}'}$ and $\mu_{\u{A}'}$. Taking $\mu_{\u{A}'} =\lambda_{\u{A}'}$ this shows that the Hermitian matrix $P_{\rm ren}^{\u{AA}'}$ is positive semi-definite, which implies that is causal and future-directed (or zero). Hence $\varpi^2 \ge 0$. 

If $Q=P=0$ then \eqref{masscharge} is trivial so assume $Q-iP \ne 0$. $P_{\rm ren}^{\u{AA}'}$ cannot vanish for otherwise we could choose $\lambda_{\u{A}'}$ and $\mu_{\u{A}'}$ to make the RHS of \eqref{ISQP} negative. Now extremize \eqref{ISQP} w.r.t. $\mu_{\u{A}'}$ to obtain
\be
\label{lambda_sol}
 \lambda_{\u{A}'} = -\frac{\sqrt{2}}{Q-iP}\bar{\epsilon}_{\u{A'B'}} P^{\u{BB'}}_{\rm ren} \bar{\mu}_{\u{B}}
\ee
Substituting this back into \eqref{ISQP} and performing some matrix algebra gives
\be
\label{ISQP2}
 \frac{1}{4\pi} \hat{I}_S[\epsilon] = 
 \left( \frac{\varpi^2}{Q^2+P^2} -1 \right) P_{\rm ren}^{\u{AA}'} \bar{\mu}_{\u{A}} \mu_{\u{A}'}
\ee
The LHS is non-negative and $P_{\rm ren}^{\u{AA}'}$ is non-zero and positive semi-definite. Hence the expression in brackets is non-negative, i.e., \eqref{masscharge} must be satisfied. 
 
 Now assume that \eqref{masscharge} is saturated. If $Q=P=0$ then this gives $\varpi=0$ which implies that $P_{\rm ren}^{\u{AA}'}$ is degenerate and hence has a non-trivial kernel. Let $\lambda^{(0)}_{\u{A}'}$ be a member of this kernel. Then taking $\mu_{\u{A}'}$ proportional to $\lambda^{(0)}_{\u{A}'}$ in \eqref{ISQP} shows that $\hat{I}_S[\epsilon]$ has a non-trivial kernel. Similarly if $Q-iP \ne 0$ then for any $\mu_{\u{A}'}$ if we impose \eqref{lambda_sol} then (from \eqref{ISQP2}) we get an element of the kernel of $\hat{I}_S[\epsilon]$. (In both cases, the kernel has at least 2 complex dimensions.) From Theorem \ref{thm:Ipos} we deduce that (i),(ii),(iii) of that theorem must hold.  
 \end{proof}

A symmetry $2$-sphere in the extremal Reissner-Nordstr\"om spacetime is an example of a surface for which \eqref{masscharge} is saturated. The theorem is not as sharp as one would like because it does not establish that (i), (ii) and (iii) are {\it sufficient} for \eqref{masscharge} to be an equality. This is because if $S$ is a supersymmetric surface then it is not clear that the components $\lambda_{1'}$ and $\mu_{1'}$ of the associated spinor will satisfy \eqref{DM2a} on $S$. In other words there may exist supersymmetric surfaces that do not saturate \eqref{masscharge}. One might think that we could overcome this problem if we replaced \eqref{DM2a} with its supercovariant modification, obtained by substituting $\lambda_{0'}$ and $\mu_{0'}$ given by \eqref{hol_conds1} into \eqref{hol_conds2}. However this gives {\it coupled} equations for $\lambda_{1'}$ and $\mu_{1'}$, so we can't define $\lambda_{1'}^{\u{A}'}$ as before and our definition of $P_{\rm ren}^{\u{AA}'}$ no longer makes sense.  

 \begin{corollary}
 With the same assumptions as theorem \ref{thm:masscharge}, if $S$ is a trapped surface then the inequality \eqref{masscharge} is strict.
 \end{corollary}
 \begin{proof}
 If the inequality \eqref{masscharge} is saturated then theorem \ref{thm:masscharge} tells us (point (i)) that $S$ is a supersymmetric surface. Lemma \ref{lem:trapped} tells us that $S$ is not trapped.
 \end{proof}
  
\section{Discussion}

\label{sec:discuss}

We have proved a third law of black hole mechanics for supersymmetric black holes in Einstein-Maxwell theory coupled to charged matter obeying the local mass-charge inequality. This third law asserts that such black holes cannot have compact interior and hence cannot be formed in gravitational collapse. This result leaves open the possibility that a supersymmetric black hole might form in {\it infinite time} e.g. gravitational collapse might produce a black hole that approaches extremal Reissner-Nordstr\"om asymptotically at late time. Such solutions could be ``critical'' in the sense discussed by Kehle and Unger \cite{Kehle:2024vyt}, i.e., separating solutions that form black holes from solutions that are future causally geodesically complete. It would be interesting to know whether such asymptotically extremal solutions exist for matter that obeys the local mass-charge inequality.\footnote{Asymptotically extremal Reissner-Nordstr\"om solutions have been constructed numerically for uncharged scalar field matter \cite{Murata:2013daa}, which obviously obeys the local mass-charge inequality. However, charged black holes cannot form in gravitational collapse in this model.}

Our third law holds only for supersymmetric black holes. But not all extremal black holes are supersymmetric e.g. extremal Kerr is not. Kehle and Unger have conjectured \cite{Kehle:2022uvc,Kehle:2023eni} that it is possible to form an extremal Kerr black hole from regular {\it vacuum} initial data (i.e. gravitational collapse of gravitational waves). If this conjecture is correct then, since it does not involve matter, it will hold also for the class of theories that we have considered. So the third law would be violated by extremal Kerr black holes even if it is not violated by extremal Reissner-Nordstr\"om black holes. 

There is a different version of the third law of black hole mechanics which asserts that the entropy of a black hole should vanish at extremality. Of course this is violated by classical black hole solutions but it has been argued recently that quantum corrections become large near extremality and have the effect of enforcing this version of the third law \cite{Iliesiu:2020qvm}. However, this effect occurs only for {\it non}-supersymmetric black holes, so supersymmetric black holes still violate this third law \cite{Heydeman:2020hhw}. If the conjecture mentioned in the previous paragraph is correct then this is exactly opposite to the situation for the version of the third law discussed in this paper!

Our results could be generalized in various ways. The results of \cite{Gibbons:1982jg} suggest that theorem \ref{thm:Ipos} might be generalized to allow $\Sigma$ to have additional boundaries at apparent horizons. It also seems likely that our results can be adapted to the case of a negative cosmological constant, through appropriate modification of the supercovariant derivative \cite{Gibbons:1982jg}, and thereby prove a third law for supersymmetric anti-de Sitter black holes. It would also be interesting to investigate extending our results to higher dimensions.     

 \section*{Acknowledgments}
 
 I am very grateful to Piotr Chrusciel, Christoph Kehle, James Lucietti, Lionel Mason, Paul Tod, Claude Warnick and especially Ryan Unger for helpful discussions. This work was partially supported by STFC grants ST/T000694/1 and ST/X000664/1 and by the Institut Henri Poincar\'e (UAR 839 CNRS-Sorbonne Universit\'e), and LabEx CARMIN (ANR-10-LABX-59-01).

\appendix

\section{Existence of a solution of the Sen-Witten equation}

\label{app:spinor}

The proof of Theorem \ref{thm:Ipos} assumes the existence of a solution $\tilde{\epsilon}$ of the (supercovariant) Sen-Witten equation \eqref{senwitten} satisfying the boundary conditions \eqref{wittenbcs}. Here we will present a standard argument to justify this assumption. 

First note, that if a solution of this problem exists then it is unique. The argument is essentially the same as in  \cite{Gibbons:1982fy}: if one had two such solutions then one can take their difference to obtain a solution with $ \tilde{\lambda}_{1'} = \tilde{\mu}_{1'} =0$. Since $\rho'>0$, equation \eqref{hatIS} then gives $\hat{I}_S[\tilde{\epsilon}]  \le 0$ with equality if, and only if, $\tilde{\epsilon}$ vanishes on $S$. However, following the argument in the proof of Theorem \ref{thm:Ipos}, $\hat{I}_S[\tilde{\epsilon}]$ can be converted to a manifestly non-negative integral over $\Sigma$ which is a sum of two non-negative terms, one of which is the norm of $h^a_b \hat{\nabla}_a \tilde{\epsilon}$ (this is the Witten identity). Hence it follows that $\tilde{\epsilon}$ must vanish on $S$ and $\tilde{\epsilon}$ satisfies $h^a_b \hat{\nabla}_a \tilde{\epsilon}=0$ on $\Sigma$. The argument of Lemma \ref{lem:nonzero} then gives $\tilde{\epsilon} \equiv 0$ on $\Sigma$, establishing uniqueness. 

To establish existence requires more work. As noted in a footnote of \cite{Dougan:1991zz}, one has to check that the adjoint problem, with the adjoint boundary conditions, has trivial kernel. The reason for this can be understood heuristically as follows. First convert the inhomogeneous boundary conditions \eqref{wittenbcs} into homogeneous ones by shifting $\tilde{\epsilon} \rightarrow \tilde{\epsilon} + \tilde{\epsilon}_S$ where $\tilde{\epsilon}_S$ is any smooth spinor on $\Sigma$ satisfying \eqref{wittenbcs}. This converts the Sen-Witten equation into an inhomogeneous equation with homogeneous boundary conditions:
\be
\label{sw_inhom}
\hat{\cal D} \tilde{\epsilon} \equiv \gamma^b h^a_b \hat{\nabla}_a \tilde{\epsilon}=f \qquad \qquad \tilde{\lambda}_{1'} = \tilde{\mu}_{1'} =0 \; {\rm on} \; S
\ee
with $f = -\gamma^b h^a_b \hat{\nabla}_a \tilde{\epsilon}_S$. Now introduce an inner product for Dirac spinors on $\Sigma$ defined by
\be
 (\eta,\epsilon) = \int_\Sigma \eta^\dagger \epsilon
\ee
To solve \eqref{sw_inhom} we choose $\tilde{\epsilon}$ to minimize $||\hat{\cal D} \tilde{\epsilon} -f||^2$ where $||\ldots ||$ is the norm defined by our inner product. Taking the variation w.r.t. $\tilde{\epsilon}$ and integrating by parts this shows that $\eta \equiv \hat{\cal D} \tilde{\epsilon} -f$ must satisfy
\be
\label{adjoint}
 \hat{\cal D}^\dagger \eta = 0 
\ee
with $\eta^\dagger M_a \gamma^a \delta \tilde{\epsilon}=0$ on $S$ for any $\delta \tilde{\epsilon}=0$ satisfying the boundary conditions of \eqref{sw_inhom}, where $M_a$ is a unit normal to $S$ within $\Sigma$. Choosing an orthonormal basis $\{e_\mu\}$ such that $M_a$ points in the $3$-direction and $e_0$ is normal to $\Sigma$, the boundary condition on $\tilde{\epsilon}$ can be written $(\gamma^0 \pm \gamma^3)\tilde{\epsilon}=0$ for some choice of the sign $\pm$ (this follows e.g. from results in \cite{Gibbons:1982jg}). We need $\eta^\dagger \gamma^3 \delta \tilde{\epsilon}$ to vanish whenever $\delta \tilde{\epsilon}$ satisfies this boundary condition. Writing
$
 \eta^\dagger \gamma^3 \delta \tilde{\epsilon} = \frac{1}{2}  \eta^\dagger [(\gamma^0 + \gamma^3) - (\gamma^0-\gamma^3) ]\delta \tilde{\epsilon} 
$
show that this is equivalent to $(\gamma^0 \pm \gamma^3)\eta=0$, i.e., the boundary condition is self-adjoint. 
 
We now need to show that, with this boundary condition, \eqref{adjoint} admits only the trivial solution $\eta \equiv 0$, which then implies that $\tilde{\epsilon}$ satisfies \eqref{sw_inhom}. For \cite{Dougan:1991zz} this was straightforward because the operator in that case was self-adjoint, so $\eta \equiv 0$ follows from exactly the same argument used above to establish uniqueness of $\tilde{\epsilon}$. However, this argument does not work with an electromagnetic field because $\hat{\cal D}$ is not self-adjoint. In more detail we have \cite{Bartnik:2003yg}
\be
\hat {\cal D} = \gamma^i D_i - \frac{1}{2} K \gamma^0 + \frac{1}{2} E_i \gamma^i \gamma^0 -\frac{1}{4} \epsilon_{ijk} B_i \gamma^j \gamma^k
\ee
where $D$ is the intrinsic covariant derivative on $\Sigma$, $K = K^a_a$ where $K_{ab}$ is the extrinsic curvature of $\Sigma$ and $E_i$ and $B_i$ are the electric and magnetic fields on $\Sigma$. Using the fact that $\gamma^0$ is hermitian and $\gamma^i$ is antihermitian we see that taking the adjoint preserves the first three terms but reverses the sign of the final term. Hence $\hat{\cal D}^\dagger$ differs from $\hat{\cal D}$ by reversing the magnetic field on $\Sigma$ \cite{Bartnik:2003yg}. Only if the magnetic field on $\Sigma$ vanishes is $\hat {\cal D}$ self-adjoint (so existence in this case is proved). 

Fortunately there is a simple way around this problem in the case of interest to us, i.e., where $\Sigma$ is a compact surface with boundary $S$.\footnote{I am grateful to Ryan Unger for this argument.}
 We can view $\hat{\cal D}$ as a map from a certain Hilbert space (a Sobolev space $H^1$ of spinors satisfying the homogeneous boundary condition of \eqref{sw_inhom}) to another Hilbert space (square integrable spinors on $\Sigma$). If we write $\hat{\cal D} = {\cal D} + \Lambda$ (where ${\cal D}$ is the usual Sen-Witten operator and $\Lambda$ is linear in the Maxwell field) then compactness of $\Sigma$ implies that, viewed as an operator mapping between these spaces, $\Lambda$ is compact. One can then apply a standard result of Fredholm theory to deduce that the index (dimension of kernel minus dimension of cokernel) of $\hat{\cal D}$ is the same as that of ${\cal D}$, which vanishes by the above arguments. Since we have shown that $\hat{\cal D}$ has trivial kernel, it follows that it also has trivial cokernel, which establishes existence of a (unique) solution of \eqref{sw_inhom}.  

This argument works because $\Sigma$ is compact (with boundary $S$). The argument of \cite{Gibbons:1982fy} also assumes the existence of a solution of \eqref{senwitten}, but in this case $\Sigma$ is non-compact, which makes things more complicated. Theorem 11.9 of  \cite{Bartnik:2003yg} claims to justify this assumption. The approach is to apply the Witten identity to $\hat{\cal D}^\dagger$ in the same way that it was used for $\hat{\cal D}$ to establish uniqueness of $\tilde{\epsilon}$. To do this, one has to know that the contribution to this identity coming from the energy-momentum tensor and currents is non-negative. These contributions arise from using the Einstein-Maxwell constraint equations to eliminate certain terms involving the Einstein tensor and Maxwell field. However, reversing the magnetic field on $\Sigma$ changes these terms, so the energy-momentum tensor and currents appearing in the Witten identity for $\hat{\cal D}^\dagger$ are no longer the same as those in the original ``physical'' spacetime. This is not a problem provided that they still satisfy the local mass charge inequality \eqref{local_mass_charge}, which guarantees that they have a good sign. The proof in \cite{Bartnik:2003yg} appears to assume that reversing the sign of $B_i$ simply reverses the sign of $\tilde{J}_0$ and possibly also $T_{0i}^{(m)}$ (in the basis adapted to $\Sigma$), which would preserve the local mass-charge inequality. Unfortunately this seems to be incorrect: the momentum constraint takes the form \cite{Bartnik:2003yg}
\be
 T_{0i}^{(m)} \propto 2 D_j (K^i_j - K \delta^i_j) + 4 \epsilon_{ijk} E_j B_k
\ee
so $T_{0i}^{(m)}$ does not transform simply under a reversal of the sign of $B_i$. One could overcome this difficulty simply by adding an extra assumption that not only should matter respect the local mass charge inequality in the physical spacetime, but also in the unphysical spacetime resulting from reversing $B_i$ (which amounts to shifting $T_{0i}^{(m)}$ by a multiple of  $\epsilon_{ijk} E_j B_k$). However this seems unsatisfactory. Note that this problem does not arise if the Poynting vector vanishes, i.e., if the electric and magnetic fields are parallel on $\Sigma$ so in this case the existence proof is valid. 

\section{Proof of Lemma \ref{lem:charged_dust}}

\label{app:dust}

We will prove the result holds at an arbitrary point $p$. 

First consider the case where $X^a$ is null at $p$ so $V=0$ at $p$. Then \eqref{Sigma_matter} reduces to $N^a T_{ab}^{(m)} X^b=0$ and so $T_{ab}^{(m)} X^b$ must be spacelike or zero. But  it cannot be spacelike because $T_{ab}^{(m)}$ satisfies the dominant energy condition (DEC). Hence $T_{ab}^{(m)} X^b=0$. Contract this equation with an arbitrary causal vector $Y^a$ to deduce that $T_{ab}^{(m)} Y^a$ is either spacelike or null and proportional to $X_b$. But it cannot be spacelike (DEC) so $T_{ab}^{(m)} Y^a = \omega X_b$ for some $\omega$. Linearity gives $\omega = \omega_a Y^a$ so $T^{(m)}_{ab} = \omega_a X_b$. Antisymmetrizing gives $\omega_a = \chi X_a$ for some $\chi$. The weak energy condition (first equation of \eqref{local_mass_charge}) gives $\chi \ge 0$. The second equation of \eqref{local_mass_charge} now has vanishing LHS so $J_a Z^a = \tilde{J}_a Z^a = 0$ for arbitrary causal $Z^a$ hence $J_a  = \tilde{J}_a  = 0$ and \eqref{Sigma_matter2} is proved. 

Next consider the case where $X^a$ is timelike at $p$. Let $i_N T^b = N^a T^{(m)b}_a$, which is causal or zero (DEC). Also write $V=V_1 + i V_2$. 
From \eqref{Sigma_matter} we obtain 
\be
\label{JJtsq}
 \left(V_1  J \cdot N - V_2 \tilde{J} \cdot N \right)^2 = \left( X \cdot i_N T \right)^2 \ge X^2 (i_N T)^2 \ge (V_1^2 + V_2^2) \left((J \cdot N)^2 + (\tilde{J} \cdot N)^2 \right)
\ee
where we used Cauchy-Schwarz in the first inequality and \eqref{local_mass_charge} along with $X^2 = |V|^2$ in the second inequality. Rearranging gives 
$0 \ge (V_2 J \cdot N + V_1 \tilde{J} \cdot N)^2$ and hence we can write
\be
 J \cdot N = \chi V_1 X \cdot N \qquad \tilde{J} \cdot N = -\chi V_2 X \cdot N
\ee
for some $\chi$. Substituting this back into \eqref{JJtsq} the LHS and RHS are equal so both inequalities are saturated. Saturation of the first (Cauchy-Schwarz) inequality implies $i_N T^a = \alpha X^a$ for some $\alpha$ and plugging this into \eqref{Sigma_matter} gives $\alpha = \chi X \cdot N$ so we've shown
\be
 T^{(m)}_{ab} N^b = \chi X \cdot N X_a
\ee
Contracting with $N^a$ and using the weak energy condition (first inequality of \eqref{local_mass_charge}) gives $\chi \ge 0$. Now pick an orthonormal basis $\{e_0,e_i\}$ with $e_0^a = N^a$ so the above results are
\be
 T^{(m)} _{0 \mu} = \chi X_0 X_\mu \qquad J_0 = \chi V_1 X_0 \qquad \tilde{J}_0 = -\chi V_2 X_0
\ee
In this basis consider the second inequality of \eqref{local_mass_charge}. The coefficient of $(Z^0)^2$ in this equation must be zero because if we set $Z^i=0$ then we have $Z^a \propto N^a$ and the inequality becomes an equality by the above results. Next the coefficient of $Z^0 Z^i$ in this equation must vanish, for otherwise we could violate the inequality by taking $Z^0 \rightarrow \infty$ and flipping the sign of $Z^i$ if necessary. Vanishing of this coefficient gives
\be
\label{Zi}
 T^{(m)}_{00} T^{(m)}_{0i} - T^{(m)}_{0j} T^{(m)}_{ij} = J_0 J_i + \tilde{J}_0 \tilde{J}_i \qquad \Leftrightarrow \qquad \chi \left( \chi X_0^2 X_i - X_j T_{ij}^{(m)} - V_1 J_i + V_2 \tilde{J}_i \right) = 0
\ee 
The second inequality of \eqref{local_mass_charge} now reduces to (after substituting for $T^{(m)}_{0i}$)
\be
\label{ZZ}
 \left( \chi^2 X_0^2 X_i X_j - T^{(m)}_{ki} T^{(m)}_{kj}\right) Z^i Z^j \ge (J_i Z^i)^2 + (\tilde{J}_i Z^i)^2 
\ee
Choose $Z^i$ such that $X_i Z^i = 0$. Then the LHS is non-positive and the RHS is non-negative so both  must vanish. Hence we must have
\be
 T^{(m)}_{ij} = t X_i X_j \qquad J_i = j X_i \qquad \tilde{J}_i  = \tilde{j} X_i
\ee
for some $t,j, \tilde{j}$. If $X_i=0$ then we are done so assume $X_i \ne 0$. Substituting back into \eqref{ZZ}, and taking $Z^i = X^i$ gives
\be
\label{ineq1}
 \chi^2 X_0^2 - t^2 X_i X_i \ge j^2 + \tilde{j}^2
\ee
If $\chi=0$ then this gives $t=j=\tilde{j}=0$ and we are done so assume $\chi >0$.  Equation \eqref{Zi} reduces to
\be
\label{eq2}
  \chi X_0^2 - t X_i X_i = j V_1 - \tilde{j} V_2  
\ee
Using Cauchy-Schwarz on the RHS and rearranging gives
\be
 \sqrt{ j^2 + \tilde{j}^2} \ge \frac{  \chi X_0^2 - t X_i X_i}{\sqrt{V_1^2 + V_2^2}}
\ee
Substituting this into the RHS of \eqref{ineq1} and using $V_1^2 +V_2^2 = X_0^2 - X_i X_i$ gives
\be
-(t-\chi)^2 X_0^2 X_i X_i \ge 0
\ee
and hence $t=\chi$. This inequality is then saturated, which implies that the Cauchy-Schwarz inequality must be saturated so $j = \beta V_1$, $\tilde{j} = -\beta V_2$ for some $\beta$ and then \eqref{eq2} fixes $\beta = \chi$. This concludes the proof.

\end{document}